\documentclass[11pt,a4paper,reqno]{amsart}

\usepackage{geometry}

\usepackage[latin1]{inputenc} \usepackage[T1]{fontenc} \usepackage{lmodern}

\usepackage[english]{babel}

\usepackage{amsmath,amssymb,amsfonts}
\usepackage[all]{xy}

\usepackage[algoruled,linesnumbered]{algorithm2e}

\usepackage[colorlinks=true,linkcolor=blue,citecolor=blue,bookmarks=false,pdftex]{hyperref}
\hypersetup{
  pdfauthor   = {C. Dunand and R. Lercier},
  pdftitle    = {Normal Elliptic Bases and Torus-Based Cryptography},
  pdfsubject  = {},
  pdfkeywords = {}
}

\newcommand{\Trc}{\mathop{\rm Tr}}
\newcommand{\gbpgcd}{\mathop{\rm gcd}}
\newcommand{\Reslt}{\mathop{\rm Res}}

\newcommand{\mF}{\mathbb{F}}

\newcommand{\mQ}{\mathbb{Q}}

\newcommand{\mZ}{\mathbb{Z}}

\newcommand{\boM}{\mathcal{M}}

\newcommand{\FF}[1]{\mF_{#1}}
\newcommand{\FFx}[1]{\FF{#1}^{\times}}

\newcommand{\TT}[1]{T_{#1}(\FF{q})}
\newcommand{\TSet}{\mathcal T}
\newcommand{\PP}{\mathcal P}

\newcommand\agot{{\mathfrak a}}
\newcommand\bgot{{\mathfrak b}}

\newcommand{\embed}{\hookrightarrow}

\newtheorem{theorem}{Theorem}
\newtheorem{lemma}{Lemma}

\numberwithin{equation}{section}

\begin{document}

\title{Normal Elliptic Bases and Torus-Based Cryptography}

\author{Cl\'ement Dunand}%
\address{Institut de recherche math\'ematique de Rennes, Universit\'e de
  Rennes 1, Campus de Beaulieu, F-35042 Rennes Cedex, France.}%
\email{clement.dunand@univ-rennes1.fr}

\author{Reynald Lercier}%
\address{ DGA/C\'ELAR, La Roche Marguerite, F-35174 Bruz Cedex, France.  }%
\address{ Institut de recherche math\'ematique de Rennes, Universit\'e de
  Rennes 1, Campus de Beaulieu, F-35042 Rennes Cedex, France.  }%
\email{reynald.lercier@m4x.org}

\date{\today}

\begin{abstract}
  We consider representations of algebraic tori $T_n(\FF{q})$ over finite
  fields. We make use of normal elliptic bases to show that, for infinitely
  many squarefree integers $n$ and infinitely many values of $q$, we can
  encode $m$ torus elements, to a small fixed overhead and to $m$
  $\varphi(n)$-tuples of $\FF{q}$ elements, in quasi-linear time in $\log q$.

  This improves upon previously known algorithms, which all have a
  quasi-quadratic complexity.  As a result, the cost of the encoding phase is
  now negligible in Diffie-Hellman cryptographic schemes.
\end{abstract}

\maketitle


\section{Introduction}
\label{sec:introduction}

Multiplicative groups defined by finite fields $\FFx{q^n}$ are of first
importance in numerous applications, especially in discrete-log based public
key cryptography. In this field, Diffie and Hellman's seminal
paper~\cite{DiHe76} opened the way to their use in numerous cryptographic
standards in the eighties. It turns out that elliptic curves are often
prefered today, since there exist subexponential algorithms to solve the
discrete logarithm problem in finite fields~\cite{Schirokauer93}. But
$\FFx{q^n}$-subgroups of order $\Phi_n(q)$, where $\Phi_n$ denotes the $n$-th
cyclotomic polynomial (the minimal polynomial over $\mQ$ of
$e^{\frac{2i\pi}{n}}$), has reattracted attention since the publication of
Lenstra and Verheul's \textsc{xtr} scheme in 2000~\cite{LeVe00}.

Lenstra and Verheul noticed that in the very particular case $n=6$, working in
the $\FFx{q^6}$-subgroup of order $\Phi_6(q) = q^2-q+1$ can be done with a
$\FFx{q^2}$ arithmetic, whereas the best way to break the system remains to
solve discrete logarithms problems in $\FFx{q^6}$.  Certainly, this yields
reasonably competitive implementations. But the most surprising is that
\textsc{xtr} subgroups are, up to symmetry, generated by the relative trace
$\Trc_{\,\FF{q^6}/\FF{q^2}}$. As a consequence, we can encode them with only
two elements of $\FF{q}$, with time complexity equal to $\log^{1+o(1)} q$
elementary operations.  \medskip

In this paper, we exhibit for $n > 6$, $n$ fixed, encodings that can be
computed very efficiently, that is with $\log^{1+o(1)} q$ bit operations
too. To this purpose, we start from the interpretation of
\textsc{xtr}-subgroups as algebraic tori, due to Rubin and
Silverberg~\cite{RuSi03}, and the explicit encoding proposed by van Dijk and
Woodruff~\cite{DiWo04}.

Algebraic tori over $\FF{q}$ are algebraic groups defined over $\FF{q}$ that are
isomorphic to some $(G_m)^d$ over $\overline{\mF}_q$, where $G_m$ denotes the
multiplicative group and $d$ is the dimension of the torus. Algebraic tori
involved here are
\begin{equation}
  \label{eq:1}
  \TT{n} \cong \left\{ x \in \FFx{q^n}: N_{\FF{q^n}/F}(x) = 1
    \text{ whenever } \FF{q} \subset F \subsetneq \FF{q^n}, F\text{ a field}  \right\}\,.
\end{equation}
These are algebraic varieties of dimension $d=\varphi(n)$, where $\varphi$ is
the Euler-totient function. It turns out that in terms of group, $\TT{n}$
is a subgroup of order $\Phi_n(q)$, that is
\begin{math}
  \TT{n} \cong \{ x \in \FFx{q^n}: x^{\Phi_n(q)} = 1 \}\,.
\end{math}
An efficient rational parameterization of these tori with $\varphi(n)$-tuples
instead of $n$-tuples would thus allow the same security as in $\FFx{q^n}$,
but a reduced communication cost.  Even though practical constructions exist
for particular values of $n$ (for instance, 2, 3 or 6 with
\textsc{luc}~\cite{SmLe93}, \textsc{xtr}\cite{LeVe00} or
\textsc{ceilidh}\cite{RuSi03}), the rationality or stable rationality of such
structures for every $n$ has been a concern for several years
now~\cite{Voskresinskii91}.

A nice workaround proposed by van Dijk and Woodruff~\cite{DiWo04} consists in
adding to the torus $ \TT{n}$ some well chosen finite fields and mapping the
whole set into another product of finite fields,
\begin{equation}
  \label{eq:2}
  \theta: \TT{n} \times
  \prod_{\substack{d\,|\,n\\ \mu(n/d)=-1}} \FFx{q^d} \to
  \prod_{\substack{d\,|\,n\\ \mu(n/d)=+1}} \FFx{q^d}\,,
\end{equation}
where $\mu$ is the Moebius function. This bijection enables to compactly
represent $m$ elements of $\TT{n}$ with roughly $m\varphi(n)$ elements in
$\FF{q}$ for large enough $m$.  For well chosen $q$ and $n$, mainly $n$ a
product of distinct primes and $q$ of maximal order modulo these primes,
evaluating $\theta$ requires at least $n^{3+o(1)}\log^{2+o(1)} q$
elementary operations.\medskip

In the present work, we observe that the heaviest part of the complexity comes
from exponentiations in $\FF{q^n}$ to powers with sparse decomposition in
basis $q$ and we succeed in speeding up the algorithm with the help of a new
representation of field extensions. Couveignes and Lercier recently
constructed a new family of normal bases, called normal elliptic
bases~\cite{CoLe09}. They allow to perform low cost arithmetic in $\FF{q^n}$
and in the context of tori this yields encodings with a $\log q$ smaller
computational cost.  In order to reach this complexity, we need inputs $q$ and
$n$ such that $\Phi_e(q)$ and $\Phi_f(q)$ are relatively prime for any
distinct divisors $e$ and $f$ of $n$. This is not a big restriction in
applications, since there are infinitely many $n$ and $q$ such that this
condition holds.

It is worth to notice that the encoding cost becomes negligible in regard of
the major cost in many Diffie-Hellman cryptosystems, $n^{2+o(1)}\log^{2+o(1)}
q$ bit operations, due to exponentiations in $\FF{q^n}$. This is particularly
interesting since in cryptographic applications $q$ tends to be a large number
and $n$ rather small.

We may also remark that these ideas can be easily adapted to the improved
variant of $\theta$ introduced by Dijk \textit{et al.} in
2005~\cite{DiGrPaRuSiStWo05}.  They substitute tori of small dimensions for
the finite fields $\FF{q^d}$ in Eq.~\eqref{eq:2}, but all the calculations
still take place in $\FF{q^n}$ and can be sped up thanks to normal elliptic
bases.\bigskip

\noindent \textbf{Outline.}  In Section~\ref{sec:algebraic-tori}, we present
some background materials about algebraic tori encodings.
Section~\ref{sec:ellipt-peri-algebr} outlines some nice cyclotomic properties
of these algorithms and shows how the use of a normal elliptic basis can yield
a $\log q$ speedup. Section~\ref{sec:crypt-appl} discusses some of the
cryptographic applications of these mappings.

\section{Explicit Algebraic Tori Encodings}
\label{sec:algebraic-tori}

Van Dijk and Woodruff first proposed an algorithmic way to encode efficiently
a torus $\TT{n}$, modulo some small constraints on $q$ and $n$~\cite{DiWo04}.

\subsection{Principles}
\label{sec:principles}

We start from the embedding $\TT{n} \embed \FFx{q}$ and we complete both sides with
the missing parts in order to create a bijection.\medskip

From $q^n-1 =\prod_{d\,|\,n} \Phi_d(q)$, we have $\FFx{q} \simeq
\prod_{d\,|\,n} \TT{d}$. Van Dijk and Woodruff first add the product
$\prod_{d\,|\,n, d \neq n} \TT{d}$ to the left hand side of the
embedding. Then, they identify factors of the form $\FFx{q^d}$ with $d\,|\,n$
in this expression. At this step, we may have to add some newer tori, of
smaller dimension. As a result, this will modify the right hand side too. But
again, we identify there factors of the form $\FFx{q^d}$. After enough such
iterations, this yields a bijection $\theta$ (\textit{cf.}  Eq.~\eqref{eq:2}).

The domain of this bijection is much larger than $\TT{n}$, but in the case
where we have $m$ elements of $\TT{n}$ to encode, we can nevertheless recover
a quasi optimal encoding rate. We refer to Section~\ref{sec:key-agreement} for
details.  \bigskip

\noindent \textbf{Example.}
Let us see how it works for $n=15$. We have 
\begin{displaymath}
  \TT{1} \times \TT{3} \times \TT{5} \times \TT{{15}} \simeq \FFx{q^{15}}\,.
\end{displaymath}
So, $(\TT{1} \times \TT{3})\times (\TT{1} \times \TT{5}) \times \TT{15}
\simeq \FFx{q^{15}} \times \TT{1}$, hence the bijection 
\begin{displaymath}
  \FFx{q^3} \times \FFx{q^5} \times \TT{{15}} \xrightarrow{\sim}  \FFx{q^{15}} \times
  \FFx{q}\,,
\end{displaymath}
since
\begin{displaymath}
  \TT{1}\simeq \FFx{q},\ \TT{3}\times \TT{1} \simeq \FFx{q^3}\ \text{ and }\ 
  \TT{5}\times \TT{1} \simeq \FFx{q^5}.  
\end{displaymath}
Let us remark that there is no guarantee that the $\Phi_d(q)$'s are coprime,
and thus this bijection may not be a group isomorphism.

\subsection{Explicit Encodings}
\label{sec:an-explicit-version}

We now show how we can explicitly construct the bijection $\theta$. We can
obtain its inverse in the same way, but for the sake of simplicity, we omit
details.\medskip

For all $d\,|\,n$, call $U_d$ the smallest positive integer such that
\begin{equation}
  \label{defu}
  \forall e\,|\,d,\ \forall f\,|\,d \text{ with } e\ne f,\ \gcd\left(\Phi_e(q),\Phi_f(q),\frac{q^d-1}{U_d}\right) = 1.
\end{equation}
For $e\,|\,d\,|\,n$, let furthermore
$y_{d,e}=\gcd\left(\Phi_e(q),(q^d-1)/{U_d}\right)$ and
$z_{d,e}=\gcd(\Phi_e(q),$ $U_d)$.  Let finally $w_d$, $w_{d,e}$ and $u_{d,e}$,
$v_{d,e}$ be the
coefficients in B\'ezout's relations
\begin{equation} \label{eq:3}
  \frac{q^d-1}{U_d} w_d + \sum_{e\,|\,d} \frac{q^d-1}{y_{d,e}}w_{d,e} =
  1\ \text{ and }\ 
  \frac{\Phi_e(q)}{y_{d,e}} u_{d,e} + \frac{\Phi_e(q)}{z_{d,e}} v_{d,e} = 1\,.
\end{equation}
With the notations above, we have the following bijections, for all $d\,|\,n$,
\begin{displaymath} 
  \FFx{q^d} \xrightarrow{\sim} \mZ/U_d\mZ \times \prod_{e\,|\,d} \mZ/y_{d,e}\mZ
  \ \text{ and }\ 
  \mZ/U_d\mZ \xrightarrow{\sim} \prod_{e\,|\,d} \mZ/z_{d,e}\mZ\,.
\end{displaymath}
These two successive bijections give a full decomposition of each $\FF{q^d}$
into 
\begin{displaymath}
  \left(\prod_{e \,|\, d} \mZ/y_{d,e}\mZ \right) \times
  \left(\prod_{e\,|\,d} \mZ/z_{d,e}\mZ\right)\,.
\end{displaymath}

The first bijection is a canonical bijection given by the Chinese remainder
theorem, whereas the second one is non-canonical and can be performed by a
table lookup. Van Dijk and Woodruff have proved that these tables are of reasonable
size when some technical conditions are satisfied by $n$ and $q$, mainly $n$
being a product of distinct primes and $q$ of maximal order modulo these
primes.  \medskip

The idea is now to give a decomposition of both sides of the bijection
$\theta$ and to identify the small groups on each sides. The same groups
appear in a different order, except $\TT{n}$ which is mapped into
$\mZ/y_{n,n}\mZ \times \mZ/Z_{n,n}\mZ$.  For each $d\,|\,n$, $d\neq n$, we
identify $\prod_{e \,|\, d} \mZ/z_{d,e}\mZ \longrightarrow \prod_{e \,|\, d}
\mZ/z_{\rho_e(d),e}\mZ$ where $\rho_e$ is the bijection 
\begin{displaymath}
\rho_e: \{d: e\,|\,d\,|\,n,
\mu(n/d)=1\}\xrightarrow{\sim} \{d: e\,|\,d\,|\,n, \mu(n/d)=-1\}\,.
\end{displaymath}
All in all, we obtain Algorithm~\ref{algotheta}.

\begin{algorithm}[ht]
\KwIn{$x\in \TT{n}$ and $x_d\in \FFx{q^d}$ for all $d\,|\,n$ such that
$\mu(n/d)=-1$.}

\KwOut{$x_d\in \FFx{q^d}$ for all $d\,|\,n$ such that
  $\mu(n/d)=1$.\\\ \\}

\ForEach{$d\,|\,n$ {\rm such that} $\mu(n/d)=-1$}{ \small

  Compute $x_d \mapsto x_d^{(q^d-1)/U_d}$, the canonical map $\FFx{q^d} \to
  \mZ/U_d\mZ$\,.

  Compute $x_d^{(q^d-1)/U_d} \mapsto (Z_{d,e})_{e\,|\,d}$, the table lookup
  $\mZ/U_d\mZ \to \prod_{e \,|\, d} \mZ/z_{d,e}\mZ$\,.

  Map $(Z_{d,e})_{e\,|\,d} \mapsto (Z_{\rho_e(d),e})_{e\,|\,d}$ with
  $Z_{\rho_e(d),e} =
  (Z_{d,e}^{v_{d,e}}x_d^{(q^d-1)u_{d,e}/y_{d,e}})^{\Phi_e(q)/z_{\rho_e(d),e}}$,
  that is map $\prod_{e \,|\, d} \mZ/z_{d,e}\mZ \to \prod_{e \,|\, d}
  \mZ/z_{\rho_e(d),e}\mZ$\,.
}
Compute $Z_{n,n} = x^{\Phi_n(q)/z_{n,n}}\in \mZ/ z_{\rho(n),n}\mZ$.

\ForEach{$d\,|\,n$ {\rm such that} $\mu(n/d)=1$}{ \small

  Compute $(Z_{d,e})_{e\,|\,d} \mapsto Z_d$, the table lookup
  $\prod_{\substack{\rho_e(d')=d,e \,|\, d\\e\neq d}} \mZ/z_{d',e}\mZ \to
  \mZ/U_d\mZ$\,.

  Compute $x_d = Z_d^{w^d} \prod_{\substack{\rho_e(d')=d,e \,|\, d\\e\neq d}}
  (Z_{d',e}^{v_{d',e}}
  x_{d'}^{(q^{d'}-1)u_{d',e}/y_{d',e}})^{\Phi_e(q)w_{d,e}/y_{d,e}} \in
  \FFx{q^d}$.

}

Multiply $x_n$ by $x^{\Phi_n(q)w_{n,n}/y_{n,n}}$.

\caption{Computation of $\theta$.}
\label{algotheta}
\end{algorithm}
\bigskip

\noindent
\textbf{Example.} We focus again on the case $n=15$, with $U_d=1$ for all
$d\,|\,n$ which gives good insights of what actually happens. We sketch the
construction on Fig.~\ref{ex15}.

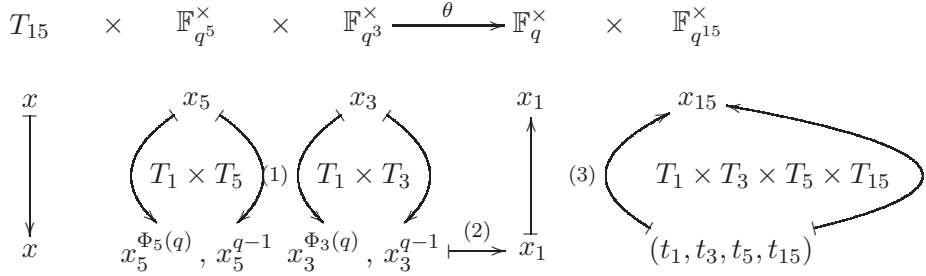
\begin{figure}[ht]
\begin{displaymath}
  \xymatrix @!0 @R=1cm @C=11.1mm {
    T_{15}&\times&\FFx{q^5}&\times&\FFx{q^3}\ar[rr]^{\theta}&&
    \FFx{q}&\times&\FFx{q^{15}}\\
    x\ar@{|->}[dd]&&x_5\ar@/_25pt/@{|->}[dd]\ar@/^25pt/@{|->}[dd]&&
    x_3\ar@/_25pt/@{|->}[dd]_{(1)}\ar@/^25pt/@{|->}[dd]&&x_1&&x_{15}\\
    &&T_1\times T_5&&T_1\times T_3&&&&
    \hspace{2cm}T_1\times T_3\times T_5\times T_{15}\\
    x&&x_5^{\Phi_5(q)}\, ,\, x_5^{q-1} &&x_3^{\Phi_3(q)}\, ,\,
    x_3^{q-1}\ar@{|->}[rr]^-{(2)}&&\,x_1\ar@{|->}[uu]&&\hspace{25pt}
    (t_1,t_3,t_5,t_{15})\ar@{|->}@/^35pt/[uu]^{(3)}\ar@{|->}@/_85pt/[uu]
  }
\end{displaymath}
\caption{The bijection $\theta$ for $n=15$ and $U_1=U_3=U_5=U_{15}=1$.}
\label{ex15}
\end{figure}

We have here several simplifications. For every $e\,|\,d$, $y_{d,e} =
\Phi_e(q)$ and $z_{d,e} =1$. Then the groups $\mZ/y_{d,e}\mZ$ involved are
nothing but the tori $\TT{e}$. Besides $u_{d,e} = 1$ and $v_{d,e} =0$.
Eq.~\eqref{eq:3} becomes
\begin{math}
  \sum_{e\,|\,d} \frac{q^d-1}{\Phi_e(q)}w_{d,e} =1\,,
\end{math}
and $x_{15}$ is simply given by $x_{15} =
t_1^{w_{15,1}}t_3^{w_{15,3}}t_5^{w_{15,5}}t_{15}^{w_{15,15}}$.

An explicit computation shows that the $w_{15,e}$'s have a convenient common
denominator, namely $15$.  So, $x_{15} =
(t_1^{r_1}t_3^{r_3}t_5^{r_5}t_{15}^{r_{15}})^{1/15}$, where the $r_e$'s are
convenient polynomials in $q$,
\begin{displaymath}
  \begin{cases}
    r_1 = 1,\\
    r_3 = -q-2,\\
    r_5 = -{q}^{3}-2\,{q}^{2}-3\,q-4,\\
    r_{15} = {q}^{7}-3\,{q}^{5}+4\,{q}^{4}-5\,{q}^{3}+7\,q-8.\\
  \end{cases}
\end{displaymath}
The cost is as follows (\textit{cf.}  Fig.~\ref{ex15}).
\begin{description}
\item[Phase (1)\,] Exponentiations to the powers $q-1$, $\Phi_3(q) = q^2+q+1$ and
  $\Phi_5(q)=q^4+q^3+q^2+q+1$ cost in average, respectively, $\frac{1}{2} \log
  q$, $\frac{1}{2} (2\log q)$ and $\frac{1}{2} (4\log q)$ multiplications
  since we perform exponentiations to power of the sizes $q$, $q^2$ and $q^4$.
\item[Phase (2)\,] Negligible.
\item[Phase (3)\,] Recall the expressions of the $r_e$'s. Exponentiation to these powers
  demands in average $\deg r_e \times (\frac{1}{2} \log q)$. So
  altogether: $(0+1+3+7) \times (\frac{1}{2} \log q)$.
\end{description}
This elementary calculation shows that, in average, the cost is about $9
\log q$ multiplications in $\FF{q^{15}}$, that is $\log^{2+o(1)} q$ elementary
operations. Van Dijk and Woodruff propose some insights to improve this cost
in practice (multi-exponentiations, redundancies, \textit{etc.}), but the
asymptotic complexity remains quasi-quadratic in $\log q$.

\subsection{Computational Complexities}
\label{sec:comp-compl}

We can now state more precisely the complexity of
Algorithm~\ref{algotheta}.\medskip

We first construct an irreducible polynomial $P(X)$ of degree $n$ over
$\FF{q}$, which can be done in $n^{2+o (1)}\log^{2+o(1)}q$
operations~\cite{PaRi98}. Let $\alpha=X \bmod P(X)$.  Then
$(1,\alpha,\hdots,\alpha^{n-1})$ is an $\FF{q}$-basis of $\FF{q^n}$.
Additions, subtractions and comparisons require $O(n\log q)$ elementary
operations. Multiplications and divisions require $n^{1+o(1)}\log^{1+o(1)} q$
elementary operations.\medskip

We also have to handle basis changes between $\FF{q^n}$ and its subfields
$\FF{q^d}$. There are $d(n)$ such subfields, where $d(n)$ is the divisor
function. This may yield large finite field lattices (see
Fig.~\ref{fig:lattice} for an example).  To simplify things, and since it does
not change the complexity, we consider that $\FF{q^d}$ elements for $d\,|\,n$
are given in the basis $(1,\alpha,\hdots,\alpha^{n-1})$ too. So, we can easily
multiply elements given in two distinct subfields. Just, in order to obtain
the right dimensions for inputs or outputs of the algorithm, we apply to an
$\FF{q^d}$ element given in $\FF{q^n}$ an $\FF{q}$-linear compression derived
from equations of the type $x^{q^d} = x$. This yields matrices $A_{n,d}\in
\boM_{n,d}(\FF{q})$ for the embedding $\FF{q^d} \embed \FF{q^n}$. Building and
applying such a matrix costs at most $n^{3}$ multiplications in
$\FF{q}$. Since there are $d(n)\simeq n^{o(1)}$ of them, this yields a total
cost of $n^{3+o(1)}\log^{1+o(1)} q$ bit operations.\medskip

Van Dijk and Woodruff outline that for ``reasonable'' integers $n$ and $q$,
mainly $n$ a product of distinct primes and $q$ of maximal order modulo these
primes, table lookup costs are negligible and the main costs are Step 4 and
Step 9 of the algorithm. They involve exponents which are derived from
cyclotomic polynomials.  Computing $\Phi_n$ can be done in time essentially
equal to its size (start from complex floating point approximations of
primitive $n$-th roots of unity and reconstruct $\Phi_n$ from these roots). We
know that this is a polynomial of degree $\varphi(n)$ with coefficients
upperbounded by $n^{d(n)/2}$~\cite{Erdoes46,Bateman49}, that is a size of at
most $n^{1+o(1)}$ bits. Evaluating all the $\Phi_d$'s at $q$ yields exponents
with $d \log q$ bits and can be done with $n^{2+o(1)}\log^{1+o(1)} q$
elementary operations. Using finally the approximate growth rate
$\sum_{d\,|\,n} d \simeq n^{1+o(1)}$, the total cost of Step 4 and Step 9 is
equal to $n^{3+o(1)}\log^{2+o(1)} q$.

\begin{figure}[ht]
  \centering
  \begin{math}
    \xymatrix @!0 @R=1.2cm @C=15mm {
      &\FF{q^{ef}}\ar@{-}[rrrr]&&&&
      \FF{q^n}\\
      \FF{q^{e}}\ar@{-}[ur]\ar@{-}[rrrr] 
      &&&&\FF{q^{ed}}\ar@{-}[ur]\\
      &\FF{q^{f}}\ar@{-}[uu]|!{[u];[ul]}\hole
      \ar@{-}[rr]|!{[dl];[uurrrr]}\hole
      \ar@{-}@/^1,8pc/[uurrrr]|!{[ul];[urrr]}\hole
      &&{}\ar@{-}[r]
      &{}\ar@{-}[r]|!{[d];[uur]}\hole
      &\FF{q^{df}}\ar@{-}[uu]\\
      \FF{q}\ar@{-}[uuur]|!{[uu];[uurrrr]}\hole\ar@{-}@/_1,8pc/[uurrrr]\ar@{-}@/^0,1pc/[uuurrrrr]|!{[uu];[uurrrr]}\hole\ar@{-}[uu]\ar@{-}[rrrr]\ar@{-}[ur]&&&&\FF{q^d}\ar@{-}[uu]\ar@{-}[ur]\ar@{-}[uuur] 
    }
  \end{math}  
  \caption{Finite field lattices for $n=def$, a product of three distinct
    primes.}
  \label{fig:lattice}
\end{figure}
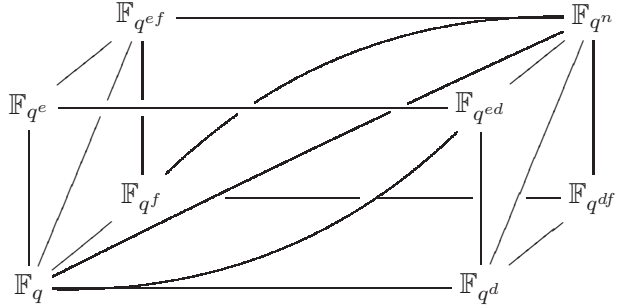

\section{Elliptic Periods and Algebraic Tori}
\label{sec:ellipt-peri-algebr}

We now focus on the case $U_d=1$ for every $d\,|\,n$.  That is no big
restriction, at least for cryptographic purposes. Indeed
Lemma~\ref{lemma:manyqandn} in Section~\ref{sec:restrictions-n-q} shows that
we can find infinitely many values of $q$ for infinitely many values of $n$
working.

We observe in Section~\ref{sec:compl} that most of the exponentiations
occuring in Algorithm~\ref{algotheta} involve exponents with a sparse
decomposition in basis $q$. This yields interests for handling $\FF{q^n}$ with
a normal basis $(\alpha,\alpha^q, \hdots,\alpha^{q^{n-1}})$ instead of a power
basis $(1,\alpha,\hdots,\alpha^{n-1})$, since with such a choice $q$-th powers
become inexpensive. Since we need to multiply elements of $\FF{q^n}$ in
quasi-linear time too, normal elliptic basis are a natural choice that we
introduce in Section~\ref{sec:norm-ellipt-bases}.

\subsection{Restrictions on $n$ and $q$}
\label{sec:restrictions-n-q}

For squarefree integers $n$, we can prove the following result.

\begin{lemma}\label{lemma:manyqandn}
  For infinitely many squarefree integers $n$, there are infinitely many
  values of $q$ such that $U_d=1$ for all $d\,|\,n$.
\end{lemma}

\begin{proof}
  From Eq.~\eqref{defu}, we deduce
  \begin{equation}
    U_d = 1 \Leftrightarrow \forall e\,|\,d,\ \forall f\,|\,d\ e\neq f,\ 
    \gbpgcd (\Phi_e(q), \Phi_f(q)) =1\,.
    \label{eq:5}
  \end{equation}
  The right hand side condition is always satisfied when $\Reslt (\Phi_e,
  \Phi_f) = 1$ and it is widely known that this is equivalent to the condition
  $f \neq e\,p^{i}$ with $p$ prime and $i\geqslant 1$ (see~\cite{Dunand09} for a
  proof). This is a corollary of the following formula due to
  Apostol~\cite{Apostol70}, for $f>e>1$,
  \begin{equation}
    \Reslt(\Phi_f,\Phi_e) = \prod_{\substack{d\,|\,e\\
        p \text{ prime},\, \frac{f}{(f,d)} = p^{i}}}
    p^{\mu(e/d)\frac{\varphi(f)}{\varphi(p^i)}}\,.
    \label{eq:4}
  \end{equation}

  There remains to check that when $f = e\,p^{i}$, there exist integers $q$
  such that Eq.~\eqref{eq:5} is satisfied. Since $n$ is supposed to be
  squarefree, the only cases are $f=ep$, $p$ prime.
  \begin{description}
  \item[Case $e=1$] The divisor $f$ is then equal to the prime $p$ and
    $\Reslt(\Phi_1,\Phi_f)=f$. In order to have $\gbpgcd (\Phi_e(q),
    \Phi_f(q)) =1$, $q$ must not be a common root of $\Phi_e$ and $\Phi_f$
    modulo $f$. In other words, we must have $q\not\equiv 1 \bmod f$.\medskip
  \item[Case $e>1$] The divisor $f$ is then equal to $pe$ where $p$ is a
    prime.  Since $e$ is squarefree, we know from Eq.~\eqref{eq:4} that
    $\Reslt (\Phi_e,\Phi_{pe}) = p^{\varphi(e)}$. So, $q$ must not be a common
    root of $\Phi_e$ and $\Phi_{pe}$ modulo $p$. Modulo $p$, $\Phi_e$ have a
    decomposition into irreducible polynomials of same degree, and this degree is
    equal to $p\bmod e$ (\textit{cf.} \cite{LiNi83}). In other words, $\Phi_e$
    and $\Phi_{pe}$ can only have a common root when $p\equiv 1 \bmod e$. In
    this case, $q$ must not be one of the $\varphi(e)$ roots of $\Phi_e$
    modulo $p$.
\end{description}

The restrictions above leave infinitely many possibilities for $q$, at least
for infinitely many values of $n$.  For instance let $n = p(p+2)$ be the
product of two twin primes and $q$ such that $q\not\equiv 1 \bmod p$ and
$q\not\equiv 1 \bmod (p+2)$.  Besides since $p+2 \not\equiv 1 \bmod p$, all
the conditions above are satisfied. Thus we have a infinite family of numbers
$q$ suitable for each $n$, and an infinite number of possible values for $n$
itself.
\end{proof}

\subsection{Normal Elliptic Basis}
\label{sec:norm-ellipt-bases}
We mimic here Couveignes and Lercier's construction.\medskip

Let $E/\FF{q}$ be an elliptic curve given by some Weierstrass model
\begin{displaymath}
  Y^2Z+a_1XYZ+a_3YZ^2=X^3+a_2X^2Z+a_4XZ^2+a_6Z^3\,.
\end{displaymath}
If $A$ is a point in $E(\FF{q})$, we denote by $\tau_A: E\rightarrow E$ the
translation by $A$.  We set $x_A=x\circ \tau_{-A}$ and $y_A=y\circ \tau_{-A}$.
If $A$, $B$ and $C$ are three pairwise distinct points in $E(\FF{q})$, we
define
\begin{displaymath}
  \Gamma(A,B,C)=\frac{y(C-A)-y(A-B)}{x(C-A)-x(A-B)}\,.
\end{displaymath}
We define a function $u_{A,B}\in \FF{q}(E)$ by $u_{A,B}(C)=\Gamma(A,B,C)$. It
has degree two with two simple poles, at $a$ and $b$.\medskip

We can prove the following identities (with Taylor expansions at poles),
\begin{equation}\label{eq:6}
  \left\{
  \begin{array}{rcl}
  \Gamma(A,B,C)&=&\Gamma(B,C,A)=-\Gamma(B,A,C)-a_1, \\
  &=&-\Gamma(-A,-B,-C)-a_1\,, \\
  u_{A,B}+u_{B,C}+u_{C,A}&=&\Gamma(A,B,C)-a_1\,,\\
  u_{A,B}u_{A,C} &=&x_A+\Gamma(A,B,C)u_{A,C}+\Gamma(A,C,B)u_{A,B}  \\
  &&+a_2
  +x_{A}(B)+x_{A}(C)\,,\\
  u_{A,B}^2&=&x_A+x_B-a_1u_{A,B}+x_A(B)+a_2\,.
\end{array}\right.
\end{equation}
  
Assume $E(\FF{q})$ contains a cyclic subgroup $\TSet$ of order $n$ and let $I:
E \rightarrow E'$ be the degree $n$ cyclic isogeny with kernel $\TSet$, then
the quotient $E'(\FF{q})/I(E(\FF{q}))$ is isomorphic to $\TSet$.

Take $A$ in $E'(\FF{q})$ such that $A\bmod I(E(\FF{q}))$ generates this
quotient.  The fiber $\PP=I^{-1}(A)=\sum_{T\in \TSet}[B+T]$ is an irreducible
divisor.  The $n$ geometric points above $A$ are defined on a degree $n$
extension of $\FF{q}$ (and permuted by Galois action),  that is $\FF{q^n}$
is the residue extension of $\FF{q}(E)$ at $\PP$.\medskip

For $k\in \mZ/n\mZ$, we set $u_k=\agot u_{kT,(k+1)T}+\bgot.$ ($\agot$ and
$\bgot$, constants chosen such that $\sum u_k = 1$).  Then the system $\Theta
= (u_k(B))_{k \in \mZ/n\mZ}$ is an $\FF{q}$ normal basis of $\FF{q^n}$.\medskip

Furthermore, there exists an algorithm with quasi-linear complexity to
multiply two elements given in an elliptic normal basis, mostly based on
Eq.~\eqref{eq:6}.  It consists in evaluations and interpolations at $d$ points
$R+kT$, where $R \in E(\FF{q})-E[n]$\,.\medskip

All of these yields Theorem~\ref{th:cole09}.

\begin{theorem}[\cite{CoLe09}]\label{th:cole09}
  To every couple $(q,n)$ with $q$ a prime power and $n\geqslant 2$ an integer
  such that $n_q \leqslant \sqrt{q}$, one can associate a normal basis
  $\Theta(q,n)$ of the degree $n$ extension of $\FF{q}$ such that the
  following holds.
  \begin{itemize}
  \item There exists an algorithm that multiplies two elements given in
    $\Theta(q,n)$ at the expense of $n^{1+o(1)}\log^{1+o(1)} q$ elementary
    operations.
  \end{itemize}
\end{theorem}
\noindent
Here $n_q$ is such that
\begin{itemize}
\item $v_\ell(n_q)=v_\ell(n)$ if $\ell$ is prime to $q-1$,
  $v_\ell(n_q)=0$ if $v_\ell(n)=0$,
\item $v_\ell(n_q)=\max(2v_\ell(q-1)+1, 2v_\ell(n))$ if $\ell$
  divides both $q-1$ and $n$.
\end{itemize}\bigskip

This result can be easily extended to a result without any restriction on
$q$ and $n$ (see \cite{CoLe09}).

\subsection{Van Dijk and Woodruff's Encoding Revisited}
\label{sec:compl}

Since $U_d=1$ for all $d\,|\,n$, van Dijk and Woodruf's encoding can be
slightly simplified. It is not only a bijection, but also a group
isomorphism.\medskip

For every $e\,|\,d$, $y_{d,e} = \Phi_e(q)$ and $z_{d,e} =1$. Then the groups
$\mZ/y_{d,e}\mZ$ involved are nothing but the tori $\TT{e}$. Besides $u_{d,e}
= 1$ and $v_{d,e} =0$.  So most of Algorithm~\ref{algotheta} is reduced to two
main phases: the decomposition $\FFx{q^d} \to \prod_{e\,|\,d} \TT{e}$ for $d$
any divisor of $n$ such that $\mu(n/d)=-1$ on the left hand side and the
reconstruction $\prod_{e\,|\,d} \TT{e} \to \FFx{q^{d}}$ for $d$ any divisor
of $n$ such that $\mu(n/d)=1$ on the right hand side.\medskip

Now we need to know what we gain with a normal elliptic basis. Essentially, it
makes each exponentiation to a power of $q$ be a simple permutation of the
basis. We thus gain a $\log q$ factor for each exponentiation of this type. It
is not difficult to see that the exponents occuring in the decomposition phase
have a sparse decomposition in basis $q$ since they are products of
evaluations of cyclotomic polynomials at $q$. But the reconstruction phase is
more tricky because it involves exponentiations by B\'ezout's coefficients
$w_{d,e}$ which do not have such a nice decomposition in basis $q$.  Instead, we
prefer to compute B\'ezout's polynomials $W_{d,e}$ such that
\begin{displaymath}
  \sum_{e\,|\,d} \frac{X^d-1}{\Phi_e(X)}W_{d,e}(X) = 1\,.
\end{displaymath}
Of course, $w_{d,e} = W_{d,e}(q) \bmod \Phi_e(q)$ \,. 

Unlike cyclotomic polynomials, these polynomials do not have integer
coefficients, but for squarefree integers $n$, and thus squarefree divisors
$d$, all their coefficients have a common denominator, equal to $d$.  More
precisely, we have
\begin{equation}\label{eq:8}
 W_{d,e}(X) = \prod_{f\,|\,d, f\neq e} \Phi_f(X)^{-1} \bmod \Phi_e(X)\,. 
\end{equation}
We may notice on the first hand that $\Phi_f(X)^{-1} \bmod \Phi_e(X)$ has got
integer coefficients if and only if $f \neq e\,p^{i}$ with $p$ prime and
$i\geqslant 1$, since $\Reslt (\Phi_e, \Phi_f) = 1$ in that case (see proof of
Lemma~\ref{lemma:manyqandn}). On the other hand, when $f = e\,p^{i}$, the
coefficients of $\Phi_f(X)^{-1} \bmod \Phi_e(X)$ have a common denominator, equal
to $f$. From Eq.~\eqref{eq:8}, and from the squarefree property satisfied by
$d$, we deduce thus that the coefficients of $W_{d,e}(X)$ have a common
denominator exactly equal to $d$. 

We observed that the numerators $R_{d,e}$ of the $W_{d,e}$'s have small
coefficients too (see Section~\ref{sec:case-n=pr-with} for a detailed analysis
in the case $n=pr$).  Consequently, we restrict $q$ to prime powers such that
$n$ is invertible modulo $q^n-1$ and slightly modify $\theta$ to output
$x_d^{n}$ instead of $x_d$ for each $d\,|\,n$ such that $\mu(n/d)=1$.  We
denote $\widetilde\theta$ this variant (\textit{cf.}
Algorithm~\ref{algothetatilde}).
\begin{algorithm}[ht]
\KwIn{$x\in \TT{n}$ and $x_d\in \FFx{q^d}$ for all $d\,|\,n$ such that
$\mu(n/d)=-1$.}

\KwOut{$x_d\in \FFx{q^d}$ for all $d\,|\,n$ such that
  $\mu(n/d)=1$.\\\ \\}

\ForEach{$d\,|\,n$ {\rm such that} $\mu(n/d)=-1$}{

  Compute $x_d \mapsto (Z_{\rho_e(d),e})_{e\,|\,d}$ with $Z_{\rho_e(d),e} =
  x_d^{(q^d-1)/\Phi_e(q)}$.

}

Set $Z_{n,n} = x $.

\ForEach{$d\,|\,n$ {\rm such that} $\mu(n/d)=1$}{

  Compute $x_d = \prod_{\substack{\rho_e(d')=d,e \,|\, d\\e\neq d}}
  Z_{d',e}^{n W_{d,e}(q)} \in \FFx{q^d}$.

}

\caption{Computation of $\widetilde\theta$.}
\label{algothetatilde}
\end{algorithm}
\bigskip

Fortunately, we do not need any more compression matrices $A_{n,d}$ with
normal basis (\textit{cf.} Section~\ref{sec:comp-compl}). In truth, a
$\FF{q^d}$ element has got a periodic set of components in any normal basis of
$\FF{q^n}$. Consequently, compressing simply consists in truncating to the $d$
first components and expanding consists in concatenating $n/d$ copies of a
$d$-tuple of $\FF{q}$ elements. Costs are negligible. \medskip

Before considering in detail the case $n=pr$ a product of two primes in
Section~\ref{sec:case-n=pr-with}, and discuss the general case in
Section~\ref{sec:case-integers-n}, we focus on an explicit example, namely
$n=15$ in order to compare with
Section~\ref{sec:an-explicit-version}.\bigskip

\noindent \textbf{Example.}  Recall Fig.~\ref{ex15} for the notations, the
costs are the following.
\begin{description}
\item[Phase (1)\,] Exponentiations to the powers $\Phi_3(q) = q^2+q+1$ and
  $\Phi_5(q)=q^4+q^3+q^2+q+1$ cost respectively 2 and 4 multiplications since
  exponentiation to a power of $q$ is free (mere permutation of the
  basis). Exponentiation to the power $q-1$ costs an inversion, which is
  performed in linear time.
\item[Phase (2)\,] Negligible.
\item[Phase (3)\,] Recall the expressions of the $r_e$'s. For instance $r_{15}
  = {q}^{7}-3\,{q}^{5}+4\,{q}^{4}-5\,{q}^{3}+7\,q-8$. Exponentiation to this
  power demands $6\times 3$ multiplications for the coefficients (6
  coefficients of size at most $2^3$) and 6 multiplications to add the 7
  monomials. The same calculation for each $r_e$ gives the global cost of
  Phase~(3): $3 + ((0)+(1\times1+1) + (2\times 2+2) + (6\times3+6))$
  multiplications and 3 inversions.
\end{description}

If we remind the total found for computations without normal elliptic bases,
it is a clear practical improvement.  The most important is that
asymptotically, the $\log q$ factor vanishes.\medskip

\subsubsection{Case $n=pr$ with $p,r$ distinct primes}
\label{sec:case-n=pr-with}

In the case $n=pr$ with $p,r$ distinct primes, the situation is very similar
to our $n=15$ example (\textit{cf.} Fig.~\ref{fig:case_pr}).\medskip

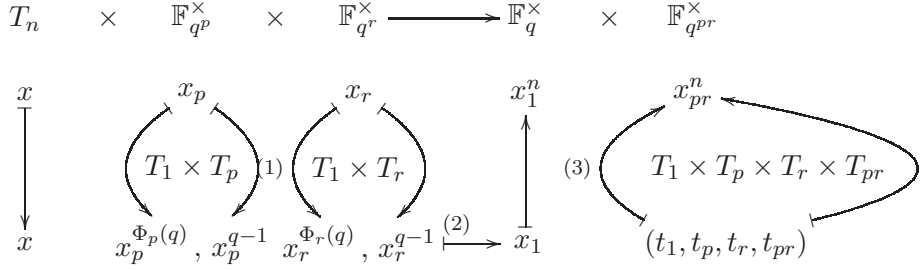
\begin{figure}[ht]
\[
\xymatrix @!0 @R=1cm @C=11.1mm {
T_n&\times&\FFx{q^p}&\times&\FFx{q^r}\ar[rr]&&
\FFx{q}&\times&\FFx{q^{pr}}\\
x\ar@{|->}[dd]&&x_p\ar@/_25pt/@{|->}[dd]\ar@/^25pt/@{|->}[dd]&&
x_r\ar@/_25pt/@{|->}[dd]_{(1)}\ar@/^25pt/@{|->}[dd]&&x_1^n&&x_{pr}^n\\
&&T_1\times T_p&&T_1\times T_r&&&&
\hspace{2cm}T_1\times T_p\times T_r\times T_{pr}\\
x&&x_p^{\Phi_p(q)}\, ,\, x_p^{q-1} &&x_r^{\Phi_r(q)}\, ,\, x_r^{q-1}\ar@{|->}[rr]^{\quad(2)}&&\,x_1\ar@{|->}[uu]&&\hspace{25pt} (t_1,t_p,t_r,t_{pr})\ar@{|->}@/^35pt/[uu]^{(3)}\ar@{|->}@/_85pt/[uu]
}
\]
\label{thetapr}
\caption{The bijection $\widetilde\theta$ for $n=pr$ and $U_1=U_p=U_r=U_{pr}=1$.}\label{fig:case_pr}
\end{figure}

Especially, the cost of Phase (1) comes from exponentiations to the powers
$\Phi_p(q)$ and $\Phi_r(q)$, that is $p$ and $r$ multiplications since
exponentiation to a power of $q$ is free. This costs $n^{2+o(1)}\log^{1+o(1)}
q$ bit operations. Exponentiation to the power $q-1$ costs an inversion, which
is asymptotically performed in quasi-linear time.\medskip

We now give details on the cost of Phase (3). We perform the embedding in two
steps. First, we combine $t_1$ and $t_{pr}$ on one hand and $t_p$ and $t_{r}$
on the other hand. Then, we combine the two results again to form the element
$x_{pr}$.  We summarize this process on Fig.~\ref{fig:reconstruct}.
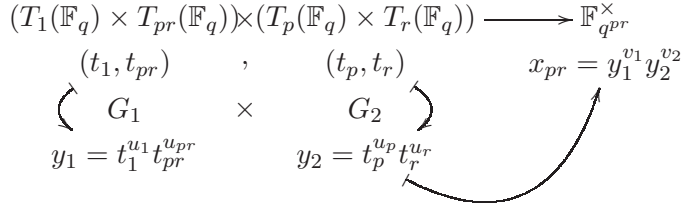
\begin{figure}[ht]
  \centering $\xymatrix @!0 @R=0.6cm @C=16mm {
    (\TT{1}\times \TT{{pr}})&\times &(\TT{p}\times \TT{r})\ar[rr]&&\FFx{q^{pr}}\\
    (t_1,t_{pr})\ar@{|->}@/_25pt/[dd]&,&(t_p,t_{r})\ar@{|->}@/^25pt/[dd]&&x_{pr} = y_1^{v_1}y_2^{v_2}\\
    G_1&\times&G_2\\y_1 = t_1^{u_1}t_{pr}^{u_{pr}}&&y_2
    =t_p^{u_p}t_{r}^{u_{r}}\ar@{|->}@/_35pt/[uurr] }$
  \caption{Reconstruction step in the case $n=pr$.}
  \label{fig:reconstruct}
\end{figure}

So the first step consists in two mappings,
\begin{displaymath}
\begin{array}[h]{rcl}
  \TT{1}\times \TT{{pr}} &\xrightarrow{\sim}& G_1 \subset \FFx{q^{pr}}\,,\\
  (t_1,t_{pr})&\mapsto & y_1 = t_1^{u_1}t_{pr}^{u_{pr}}\,,\end{array}\text{ where }
  \Phi_{pr}(q)u_1+\Phi_1(q)u_{pr} = 1\,
\end{displaymath}
and
\begin{displaymath}
  \begin{array}[h]{rcl}\TT{p}\times \TT{r} &\xrightarrow{\sim}& G_2 \subset \FFx{q^{pr}}\\
    (t_p,t_{r})&\mapsto & y_2 = t_p^{u_p}t_{r}^{u_{r}}\end{array}\text{ where }
  \Phi_{r}(q)u_p+\Phi_p(q)u_{r} = 1\,.
\end{displaymath}
The final recombination is 
\begin{displaymath}
  \begin{array}{rcl} G_1\times G_2&\to & \FFx{q^{pr}} \\
    (y_1,y_2) &\mapsto& y_1^{v_1}y_2^{v_2} \end{array}
  \text{ where } \, \frac{q^{pr}-1}{\Phi_1(q)\Phi_{pr}(q)} v_1+
  \frac{q^{pr}-1}{\Phi_p(q)\Phi_{r}(q)} v_2 =1\,.
\end{displaymath}
\medskip

The powers involved in the mappings of the first step, $u_1$, $u_p$, $u_r$ and
$u_{pr}$ are the evaluations in $q$ of respectively ${\Phi_{pr}^{-1}} \bmod
\Phi_1$, ${\Phi_{r}}^{-1} \bmod \Phi_p$, ${\Phi_{p}}^{-1} \bmod \Phi_r$,
${\Phi_{1}}^{-1} \bmod \Phi_{pr}$. Actually, the $n$-th cyclotomic polynomial
has small coefficients, $n^{1+o(1)}$ bits (\textit{cf.}
Section~\ref{sec:comp-compl}), and its computation can be done with
$n^{2+o(1)}$ elementary operations.

We would need similar magnitude results for modular inverses of cyclotomic
polynomials. To that end, Dunand recently found such bounds.
\begin{theorem}[\cite{Dunand09}]\label{th:clement}
  For all $p$ and $r$ distinct prime numbers,
  \begin{itemize}
  \item[(i)] ${\Phi_{p}^{-1}} \bmod \Phi_1 = {1}/{p}$ and
    ${\Phi_{1}^{-1}} \bmod \Phi_p = (-{1}/{p})(X^{p-2}+2X^{p-3} + \hdots
      +p-1)$.
  \item[(ii)] ${\Phi_{pr}^{-1}} \bmod \Phi_1 = 1$ and ${\Phi_{1}^{-1}}
      \bmod \Phi_{pr} = \sum_{i=0}^{\varphi(pr)-1} v_i X^i$ with
    $v_i\in\{-1,0,1\}$.
  \item[(iii)] ${\Phi_{pr}^{-1}} \bmod \Phi_p = \frac{1}{r}\sum_{i=0}^d
      X^i$ with $d \equiv r-1 \bmod p$ and ${\Phi_{p}^{-1}} \bmod \Phi_{pr}
      = \frac{1}{r}\sum_{i=0}^{\varphi(pr)-1} v_iX^i$ with $v_i<r$.
  \item[(iv)] ${\Phi_{p}^{-1}} \bmod \Phi_{r} = \sum_{i=0}^{\varphi(r)-1}
      v_iX^i$ with $v_i \in \{0, -1, +1\}$.
  \end{itemize}
\end{theorem}
The decomposition of $u_1$, $u_p$, $u_r$ and $u_{pr}$ in basis $q$ is very sparse, with only -1, 0, or 1 coefficients. The complexity of this step
is thus $O(n)$ multiplications and few inversions in $\FF{q^n}$, that is
$n^{2+o(1)}\log^{1+o(1)} q$ elementary operations.\medskip

The powers in the second step, $v_1$ and $v_2$, are the evaluations in $q$ of
respectively $\Phi_{p}^{-1}\Phi_r^{-1} \bmod \Phi_1\Phi_{pr}$ and
$\Phi_1^{-1}\Phi_{pr}^{-1} \bmod \Phi_p\Phi_r$.  Their computations require
the knowledge of $\Phi_{p}^{-1}$ modulo $\Phi_1$ and $\Phi_{pr}$,
$\Phi_{r}^{-1}$ modulo $\Phi_1$ and $\Phi_{pr}$, $\Phi_{1}^{-1}$ modulo
$\Phi_p$ and $ \Phi_r$ and finally $\Phi_{pr}^{-1}$ modulo $\Phi_p$ and $
\Phi_r$. To compute inverses modulo a product of two cyclotomic polynomials,
we make use of the Chinese remainder theorem.  If $\Phi = A \bmod \Phi_{pr}$
and $\Phi= B\bmod \Phi_{1}$, then
\begin{displaymath}
  \Phi  = \left(\frac{\Phi_1}{\Phi_1 \bmod \Phi_{pr}} A +
    \frac{\Phi_{pr}}{\Phi_{pr} \bmod \Phi_{1}} B\right) \bmod \Phi_1 \Phi_{pr}\,.
\end{displaymath}
And we have of course a similar formula for the second case.
This yields the following coefficient bounds (in absolute value),
\begin{multline}\label{eq:9}
  \Phi_p^{-1} \bmod \Phi_1 \Phi_{pr} = \underbrace{\Phi_1}_{\text{at most
    }1}\underbrace{(\Phi_1^{-1} \bmod \Phi_{pr})}_{\text{at most }1}
  \underbrace{(\Phi_p^{-1} \bmod \Phi_{pr})}_{= 1/p} \\
  + \underbrace{\Phi_{pr}}_{\text{at most }1}\underbrace{(\Phi_{pr}^{-1} \bmod
    \Phi_{1})}_{= 1} \underbrace{(\Phi_p^{-1} \bmod \Phi_{1})}_{\text{at most
    }r} \bmod \Phi_1 \Phi_{pr}
\end{multline}
We have such a bound for $\Phi_r^{-1} \bmod \Phi_1 \Phi_{pr}$ too (exchange
$p$ and $r$ in Eq.~\eqref{eq:9}).\medskip

Finally $v_1$ is the product of $\Phi_p^{-1}$ and $\Phi_r^{-1}$ modulo $\Phi_1
\Phi_{pr}$. The factor ${1}/{pr}$ appearing leads us to return $x_{pr}^{n}$
instead of $x_{pr}$. So the powers involved in the last step will be $nv_1$
and $nv_2$. A very quick analysis show that the coefficients of their
decomposition in basis $q$ are upperbounded in absolute value by $n^5$ and
this impacts the complexity by an additional but negligible $n^{o(1)}$
penalty. The total complexity of the reconstruction phase is thus equal to
$n^{2+o(1)}\log^{1+o(1)} q$ elementary operations.\medskip

As a conclusion, our variant of the bijection $\theta$ asymptotically costs,
for $n=pr$ the product of two primes, $n^{2+o(1)}\log^{1+o(1)} q$ elementary
operations.\medskip

\subsubsection{Case of  integers $n$ with more than two prime factors}
\label{sec:case-integers-n}

The decomposition phase is the easiest to quantify for general $n$. We have to
perform exponentiations to powers equal to cyclotomic polynomials evaluated at
$q$. Since we have at most $d(n)=n^{o(1)}$ such polynomials, since they are of
degree at most $n$ and since their coefficients have got $n^{1+o(1)}$ bits,
this yields a clear $n^{3+o(1)}\log^{1+o(1)} q$ bit complexity.

The reconstruction phase involves modular inverses of cyclotomic polynomials
and with our current knowledge, is seems very difficult to have in full
generality bounds similar to Dunand's ones in the case $n=pr$. It seems, but
we have no proof of this, that for integers $n$ with a fixed number of prime
factors, the coefficients of these cyclotomic inverses are upperbounded in
absolute value by a fixed power of $n$. And so, the reconstruction complexity
would not exceed the complexity of the decomposition phase.

For more general integers $n$, it is very hard to state something, except of
course that the complexity is no longer quasi-quadratic, but quasi-linear, in
$\log q$ .

\section{Cryptographic Applications}
\label{sec:crypt-appl}
 
In \cite{DiWo04}, van Dijk and Woodruff give several applications, including
a Diffie-Hellman-like multiple key exchange.  We show here
how this scheme can be adapted to our case.

\subsection{Key agreement}
\label{sec:key-agreement}

We denote in the following
\begin{math}
  \theta: \TT{n}\times \Pi^- \to \Pi^+\,,
\end{math}
the bijection $\theta$ initially defined by Eq.~\eqref{eq:2}. \medskip

Let us assume that Alice and Bob need to agree not on a single key but on a
sequence $(K_i)_{1\leqslant i\leqslant m}$ of keys, with a Diffie-Hellman based
system. Indeed, after having agreed on a generator $g$ of $\TT{n}$, each of the
keys will be $K_i=g^{x_iy_i}$ where $x_i$ and $y_i$ will be randomly chosen
respectively by Alice and Bob.\medskip
 
Alice computes the points $A_i = g^{x_i}$ on the torus and after having chosen
a random $S_0 \in \Pi^-$, she computes in turn $\theta(A_i,S_{i-1})=(a_i,S_i)$
for $i$ from 1 to $m$. She sends the $(a_i)_{1\leqslant i \leqslant m}$ and the last
output $S_m$ to Bob. So he can recover all the $A_i$'s by applying
$\theta^{-1}(a_i,S_i) = (A_i,S_{i-1})$ for $i$ decreasing from
$m$ to $1$. Finally the key is $K_i=A_i^{y_i}$.

In this way, $S_m$ and $a_1$, $\hdots$, $a_m$ encode $A_1$, \ldots,
$A_{m}$. This encoding is optimal except the small overhead $S_m$, that is
negligible for a large enough $m$. \medskip

Similarly, if Bob chooses $T_{0}\in \Pi^- $ and computes successively
$(b_i,T_{i})=\theta(B_i,$ $T_{{i}-1})$, he can send $(b_i)_i$ and $T_{m}$ to
Alice, who can recover $(B_i)_i$ by $(B_i,T_{{i}-1}) =
\theta^{-1}(b_i,T_{i})$, for $i$ from $m$ to 1. Then $K_i=B_i^{x_i}$ gives
the keys.

\subsection{Adaptation}
\label{sec:adaptation}

We need to modify this system since our bijection $\widetilde\theta$ is not
exactly the same. \medskip

We focus here on the case $n=pr$ but it works in the same way for more general
integers $n$. We want to use the bijection given in
Fig.~\ref{fig:case_pr}. Yet what we can efficiently calculate in the third
step is $(t_1,t_p,t_r,t_{pr})\mapsto x_{pr}^{n}$. So we are going to use the
slightly different mapping $\widetilde\theta$ and a reverse mapping
$\widetilde\theta'$,
\begin{displaymath}\setlength{\arraycolsep}{0.08cm}
  \begin{array}{rcl}
    \widetilde\theta:\, \TT{n}\times\FFx{q^p}\times \FFx{q^r} &\to& \FFx{q} \times \FFx{q^n}\,,\\
    (x,x_p,x_r)&\mapsto&(x_1^n,x_n^n)\,,
  \end{array}
\text{\,and\,}
  \begin{array}{cccl}
    \widetilde\theta':\, \FFx{q} \times \FFx{q^n}&\to& \TT{n}\times\FFx{q^p}\times \FFx{q^r}\,,\\
    (x_1,x_n) &\mapsto&(x^n,x_p^n,x_r^n)\,.
  \end{array}
\end{displaymath}
Since $\widetilde \theta' \circ \widetilde \theta (x, x_p,x_r)$ is no longer
equal to $(x, x_p,x_r)$ but to $(x^{n^2}, x_p^{n^2},x_r^{n^2})$, we cannot
make a direct use of the previous Diffie Hellman scheme. We have to
raise the output of our mappings to the $1/n$-th power instead. This can be easily done
by a straightforward exponentiation, but at cost $n^{2+o(1)}\log^{2+o(1)}q$.

It turns out that this cost can be decreased, but at the expense of an
additional constraint on $q$.
\begin{lemma}\label{sec:ninvbaseq}
  Let $n$ be an odd integer, let $q$ be a prime power such that $n$ divides
  $q+1$ and denote $k = (n-1)/2$, then
  \begin{equation}\label{eq:7}
    1/n \bmod (q^n-1) = \mu_0+\mu_1\,q+\mu_0\,q^2+\cdots+\mu_1\,q^{n-2}+\mu_0\,q^{n-1}\,,
  \end{equation}
  where
  \begin{displaymath}
    \mu_0 = \frac{k(q-1)+q}{n}\text{ and }\mu_1 = \frac{k(q-1)-1}{n}\,.
  \end{displaymath}
\end{lemma}
\begin{proof}
  We have
  \begin{multline*}
    n\, (\mu_0+\mu_1\,q+\mu_0\,q^2+\cdots+\mu_1\,q^{n-2}+\mu_0\,q^{n-1}) -1
    -k\,(q^n-1) =\\ 
    {\frac {k{q}^{n+2}+n{\it \mu_0}\,{q}^{1+n}+n{q}^{n} \left( {\it \mu_1}-k \right) -
 \left( k+1 \right) {q}^{2}-n{\it \mu_1}\,q-n{\it \mu_0}+k+1}{{q}^{2}-1}}\,.
  \end{multline*}
  The numerator of the right hand side is thus equal to
  \begin{displaymath}
    q^n(kq^2+n\mu_0q+n(\mu_1-k)) -
 \left( k+1 \right) {q}^{2}-n{\it \mu_1}\,q-n{\it \mu_0}+k+1
  \end{displaymath}
  and then we need to check that the coefficient of $q^n$ and the remaining
  part of this expression are both equal to zero with $\mu_0$ and $\mu_1$
  as given above.
\end{proof}

Raising elements of $\FF{q^n}$ to the $1/n$-th power where $1/n$ is given by
Eq.~\eqref{eq:7} can be done with $n^{1+o(1)}\log^{2+o(1)}q$ elementary
operations with a normal basis.  The global asymptotical cost of the encodings
in the key agreement is thus in this case $m$ times
$n^{2+o(1)}\log^{1+o(1)}q+n^{1+o(1)}\log^{2+o(1)}q$ bit operations. This is
smaller than $m$ times $n^{2+o(1)}\log^{2+o(1)}q$, the cost of $m$
Diffie-Hellman exponentiations.  \bigskip

\noindent
\textbf{Remark.} Computing $n$-th roots in $\FF{q^n}$ excludes even integers
$n$ in the construction, at least for odd prime powers $q$. But an easy
workaround consists in working in the quadratic residue subgroup of $\TT{1}$
and $\TT{2}$. This is equivalent to substitute $(q-1)/2$ and $(q+1)/2$ for
$\Phi_1(q)$ and $\Phi_2(q)$ everywhere in the construction of $\widetilde
\theta$. So, we are left at the end to compute $n/2$-th roots in $\FF{q^n}$
and all of these do not change the overall complexity of the scheme.

\newcommand{\etalchar}[1]{$^{#1}$}
\providecommand{\bysame}{\leavevmode\hbox to3em{\hrulefill}\thinspace}
\providecommand{\MR}{\relax\ifhmode\unskip\space\fi MR }
\providecommand{\MRhref}[2]{%
  \href{http://www.ams.org/mathscinet-getitem?mr=#1}{#2}
}
\providecommand{\href}[2]{#2}

\end{document}